\documentclass[11pt,letterpaper]{article}

\usepackage{theorem,latexsym,graphicx}
\usepackage{amsmath,amssymb,enumerate}
\usepackage{algorithm}
\usepackage{algorithmicx}
\usepackage{algpseudocode}

\usepackage{fullpage}

\usepackage[
normalsections, 
normalmargins, 
normalfloats, 
normaltitle, 
normalleading, 
normalbib, 
normalbibnotes, %
]{savetrees}

\newtheorem{definition}{Definition}[section]

\newtheorem{theorem}[definition]{Theorem}

\newtheorem{corollary}[definition]{Corollary}
\newtheorem{claim}[definition]{Claim}

\newtheorem{lemma}[definition]{Lemma}

\newenvironment{proof}{{\bf Proof:} \rm}{\hfill $\square$ \medskip\\}

\date{}
\begin{document}
\title{Combinatorial Auctions with Budgets}
\date{\today}
\author{Amos Fiat\thanks{Tel Aviv University; Research partially supported by a grant from the Israel Academy of Sciences}
\and Stefano Leonardi\thanks{Sapienza University of Rome}
\and Jared Saia \thanks{University of New Mexico;  Research was partially supported by NSF CAREER Award 0644058 and NSF CCR-0313160} \and Piotr Sankowski\footnotemark[2] \thanks{University of Warsaw} }

\begin{titlepage}
\def\thepage{}
\thispagestyle{empty} \maketitle

\maketitle

\begin{quote}
\hspace{0.5cm} {\small\raggedright \noindent {\em ``$\ldots$ any color that he [the customer] wants so long as it is black.''  }}

\hspace{3cm} {\small\raggedleft \noindent {--- Henry Ford, My Life and Work (1922) }}
\end{quote}

%

\begin{quote}
\hspace{0.5cm} {\small\raggedright \noindent {\em ``$\ldots$ Illogical approach to advertising budgets $\ldots$''  }}

\hspace{3cm} {\small \hbox{\begin{tabular}{l} --- Michael Schudson,  Advertising, The Uneasy Persuasion: \\ \qquad Its Dubious Impact on
American Society (1984) \end{tabular}}}
\end{quote}

\addvspace{0.5cm}

\begin{abstract}
We consider budget constrained combinatorial auctions where bidder $i$ has a private value $v_i$ for each of the items in some set $S_i$, agent
$i$ also has a budget constraint $b_i$. The value to agent $i$ of a set of items $R$ is $|R \cap S_i| \cdot v_i$. Such auctions capture adword
auctions, where advertisers offer a bid for those adwords that (hopefully) reach their target  audience, and advertisers also have budgets. It
is known that even if all items are identical and all budgets are public it is not possible to be truthful and efficient. Our main result is a
novel auction that runs in polynomial time, is incentive compatible, and ensures Pareto-optimality. The auction is incentive compatible with
respect to the private valuations, $v_i$, whereas the budgets, $b_i$, and the sets of interest, $S_i$, are assumed to be public knowledge. This
extends the result of Dobzinski et al. \cite{DBLP:conf/focs/DobzinskiLN08,DLNcorrected} for auctions of multiple {\sl identical} items and
public budgets to single-valued {\sl combinatorial} auctions with public budgets.

\end{abstract}

\end{titlepage}
\newpage

\section{Introduction}
In recent years ad auctions have been the subject of some non-negligible attention, perhaps because Internet ad revenue in 2009 was some $\$2.4 \times 10^{10}$ USD\footnote{According to the Interactive Advertising Bureau, www.iab.net.}. Much practical and theoretical work has been done on the issue of ad auctions, much of this work within the general framework of mechanism and auction design. If all advertisers bid for (multiple) copies of a single search term, (so called ``multi unit auction"), then  --- the Vickrey multi unit auction~\cite{vickrey-61} is both truthful and maximizes efficiency.

The Vickrey multi unit auction is not entirely satisfactory, partially because of the following:
 \begin{enumerate} \item Budgets - budgets are a necessary evil because of limited resources and risk aversion. In any real system, budgets are a key component. The Vickrey multi unit auction is not incentive compatible when budgets are allowed. Moreover, even efficiency is ill defined in such a setting (the next best thing is Pareto-optimality).
  \item Not all items are equal, typically, if you want to sell precious metals you probably want to advertise on search terms ``Gold", ``Silver", ``Platinum" and not ``Lead" or ``Corn". If you sell all metals excluding Silver and Platinum then you may want to advertise on search terms ``Gold", ``Uranium", ``Plutonium", and ``Lead".  Multiple parallel multi unit auctions, one for each and every search term, are somewhat problematic and certainly not strategy proof.
  \end{enumerate}

  One real system that addresses both issues is Google's Auction for TV Ads,
deployed a few years ago~\cite{Nisan09}. This auction allows bidders to select shows, times, and days they wish to advertise on; and then give a per-ad impression bid and a total budget. The theoretical analysis of Google's TV ad auction is yet incomplete, but it is known not to be incentive compatible --- strategic bidders can gain by misrepresentation of their valuation (the Google system does not allow one to choose different valuations for different ad slots), even if all other bidder parameters are public.

  Much of the theoretical work on mechanism design has ignored budgets. This may be because budgets mean that utilities are not quasi-linear, the Vickrey-Clarke-Groves (VCG) mechanism is not incentive-compatible, and other curiosities.

 A seminal paper on mechanisms for ad auctions with budgets is by Dobzinski, Lavi and Nisan~\cite{DBLP:conf/focs/DobzinskiLN08,DLNcorrected}. They consider multi unit auctions (all items are identical). E.g., multiple occurrences of the same search word.
  Dobzinski et. {\sl al.} give an incentive-compatible auction (with respect to valuation) that produces a Pareto-optimal allocation. This result holds if one assumes that the budgets are public information and \cite{DBLP:conf/focs/DobzinskiLN08,DLNcorrected} also show that this
  assumption is required: there is no incentive-compatible auction with respect to both valuation and budgets that produces a Pareto-optimal allocation.

  Subsequently, Aggarwal, Muthukrishnan,
Pal and Pal~\cite{unit-demand} considered the case where bidders seek at most one item --- not quite relevant for ad auctions. In this setting they give an incentive compatible auction, with respect to both valuation and budgets. This latter result is related to the paper of Hatfield and Milgrom~\cite{unit-demand2} who consider more general non-quasi-linear utilities. Both \cite{unit-demand} and \cite{unit-demand2} are in a more general combinatorial setting where agents are interested in a given subset of items, or may even can have different valuations for items.

Our work here seeks to map out the frontier of the possible. We give incentive compatible combinatorial auctions with budgets that produce Pareto-optimal allocations, for some not entirely general but also non-trivial class of auctions (the same class considered in Google's TV ad auction). Furthermore, we show that these restrictions cannot be circumvented. Thus, arguably, what we do here is the most that can be done, given that we require that the allocation is Parteo-optimal.

In this paper we study combinatorial auctions of the following {\bf general form}:
\begin{itemize}
\item Every agent (bidder) $1 \leq a \leq n$ has a {\sl publicly known} budget, $b_a\geq 0$, and an
unknown (private) valuation $v_a > 0$;
\item Every agent $a$ is ``interested" in some {\sl publicly known} set of  items, $S_a$.
 We assume that there is at least one agent interested in every
item. Agent $a$ is allocated some (possibly empty) subset of $S_a$.
\item The auction produces an allocation $(M,P)$. $M \subseteq \{1,\ldots,n\} \times \{1, \ldots, m\}$ is a (partial) matching between agents (bidders) and items. $P\in \Re^n$ is a vectors of payments made by the agents. For agent $1\leq a \leq n$, let $M_a$ be the number of items sold to agent $a$ over the
course of the auction and $P_a$ be the total payment made by agent $a$ during the course of the auction. The allocation must obey the following conditions:
\begin{enumerate}
\item The payment by agent $a$, $P_a$, cannot exceed the budget $b_a$.
\item The utility for agent $1 \leq a \leq n$ is $u_a = M_a v_{a} - P_a$.
\item The utility for the auctioneer is $\sum_{j=1}^n P_j$.
\item Bidder-rationality: for all agents $1 \leq a \leq n$, $u_a\geq 0$.
\item Auctioneer-rationality: the utility of the auctioneer, $\sum_{j=1}^n P_j \geq 0$.
\end{enumerate}
\end{itemize}
Note\footnote{In \cite{DLNcorrected} the authors refer to what we call auctioneer rationality by the term ``weakly no positive transfers".} that auctioneer-rationality is implied by  {\sl no positive transfers}: $P_a \geq 0$ for all $1 \leq a \leq n$.

Given valuations, $v_a$, budgets, $b_a$, and sets of interest, $S_a$, we define
  $(M,P)$ to be {\sl Pareto-optimal} if there is no other allocation $(M',P')$ such\footnote{Note that no restrictions are placed on the matching $M'$ or on the payments $P'$.} that
\begin{enumerate}
    \item The utility of every bidder in $(M,P)$ is not less than the utility in $(M',P')$, and
  \item The utility of the auctioneer in $(M,P)$ is not less than the utility in $(M',P')$, and
  \item At least one bidder or the auctioneer is better off in $(M',P')$ compared with $(M,P)$.
\end{enumerate}

An auction is said to be incentive compatible if it is a dominant strategy for all bidders to reveal their true valuation. An auction is said to
be Pareto-optimal if the allocation it produces is Pareto-optimal. An auction is said to make no positive transfers if the allocation it produces has no positive transfers.

 When the sets $S_a$ consist of all items for all agents, {\sl i.e.}, all items are identical, Dobzinski, Lavi, and
Nisan~\cite{DBLP:conf/focs/DobzinskiLN08,DLNcorrected} show that there are no incentive compatible mechanisms that are Pareto-optimal when both
valuations and budgets are private. Furthermore, they also show that a version of Ausubel's dynamic clinching multi-unit
auction~\cite{Ausubel04} is truthful and Pareto-optimal for agents with budgets, {\sl when budgets are public knowledge}.

\section{Our Results}

In this paper we give an incentive compatible and Parteo-optimal combinatorial auction.

Furthermore, our auction makes no positive transfers.

Our result can be viewed as extending the
results of \cite{DBLP:conf/focs/DobzinskiLN08} from selling off multiple {\sl identical} items to a new combinatorial setting where items are distinct and different agents may be interested in different items. In particular, for the non-combinatorial multi unit setting of \cite{DBLP:conf/focs/DobzinskiLN08,DLNcorrected}, our auction and the auction of \cite{DBLP:conf/focs/DobzinskiLN08,DLNcorrected} produce the same allocation. That said, we claim that our version, when restricted to the simpler multi unit setting, is much easier to follow\footnote{Karl Popper would say that this claim cannot be falsified.}.

 Our combinatorial auction is polynomial time and deterministic. Obviously, this cannot be if we were to consider the full generality of
combinatorial auctions. We consider combinatorial auctions were agents have an agent-specific set of interesting items, but only one valuation for any item from that set of interest.

In light of the impossibility results of Dobzinski et al.
\cite{DBLP:conf/focs/DobzinskiLN08} we could not hope to achieve this result with private budgets. We further show that public budgets alone are insufficient for Pareto-optimality and incentive compatibility. We prove that one cannot avoid the restrictions we place on the combinatorial auction setting in the following sense:
\begin{itemize}
\item if budgets are public but the sets of interest and the valuations are
private then no truthful Pareto-optimal auction is possible;
\item if budgets are public and private arbitrary valuations are allowed, no truthful and Pareto-optimal auction is possible (irrespective of computation time). This follows by simple reduction to the previous claim on private sets of interest.
\end{itemize}

In Section \ref{sec:dynclinch} we present our mechanism. It is straightforward to show that the mechanism is truthful with respect to valuations. However, it is not trivial to prove that the mechanism is Pareto optimal. In Section
\ref{sec:pareto-opt} we prove that the allocation produced by the mechanism is in fact Pareto optimal. In Section \ref{sec:impossible} we complement our positive result by showing that with public budgets, private valuations, and private sets of interest, there can be no truthful Pareto optimal mechanism.

\section{Combinatorial Auctions with Budgets via Dynamic Clinching}
\label{sec:dynclinch}

In this section we describe our mechanism in detail.

Our auction can be implemented as a direct revelation mechanism (where the agents reveal their private types to the
mechanism) but may also be viewed as an incentive compatible ascending auction (where incentive compatible means ex-post Nash).  The ascending
auction raises the price of unsold items till all items are clinched.  We describe the mechanism as a direct revelation mechanism and assume
that the private value $\tilde{v}_a$ is equal to the bid $v_a$. The details of the mechanism are presented in Algorithm~\ref{alg:combclinch}, Algorithm~\ref{alg:avoidmatch} and Algorithm~\ref{alg:sell}.

Throughout the algorithm there is always some current price $p$ (initially zero), current number of unsold items, $m$ (initally equal to to
total number of items), and current remaining budgets $b=(b_1, b_2, \ldots, b_n)$,
where $b_a$ is the remaining budget for agent $1 \leq a\leq n$. In addition, the algorithm maintains a boolean vector $H=(H_1, H_2, \ldots, H_n)$.

For every agent $1 \leq i\leq n$ the mechanism makes use of  values $D_i$, $D_i^+$, and $d_i$,  these values are functions of the current
values of $p$, $m$, $b$, and $H$. {\sl I.e.}, whenever one of these values is referenced it is computed based upon the current values of $p$, $m$, $b$, and $H$. Later on, we omit these arguments in the description of the mechanism. Formally:

\begin{eqnarray}
D_i = D_i[p,b_i,m] &=& \left\{\begin{array}{ll}
{\tt min} \{m,\lfloor b_i/p \rfloor\} & {\rm if\ } p\leq v_i  \\
0 & {\rm if\ } p>v_i
\end{array} \right. \label{eq:Diz}\\
 D_i^+ = D_i^+[p,b_i,m] &=&  \lim_{\epsilon \rightarrow 0^+} D_a[p+\epsilon,b_i,m];  \label{eq:Dip}\\
d_i = d_i[p,b_i,m,H] &=& \left\{\begin{array}{ll}
D_i & {\rm if\ } H_i = \rm{True} \\
D_i^+ & {\rm if\ } H_i = \rm{False}
\end{array} \right. \label{eq:dis}
 \end{eqnarray}

$D_i$ is equal to the number of items that agent $i$ is interested in purchasing at current $p$, $m$, and $b$ (Equation \ref{eq:Diz}). In
Equation \ref{eq:Dip} we define $D_i^+$, what is equal to the number of items that agent $a_i$ would be interested in purchasing if the price
were increased by an infinitesimally small amount, thus $D_i^+ \leq D_i$. In Equation (\ref{eq:dis}) we define $d_i$, {\sl the current demand of
agent $i$}, $d_i$ is either equal to $D_i$ or to $D_i^+$, depending on the value of $H_i$.

The algorithm also implicitly keeps a set of unsold items $U$ (those items not yet sold in Algorithm \ref{alg:sell}),  a set of active agents
$A$
---  those with current demand greater than zero, and a set of value limited agents $V$ --- those with valuation equal to the current price:
 \begin{eqnarray}  A &=& \{ 1 \leq  a \leq n | d_a >0\},\label{eq:A} \\
 V &=& \{ 1 \leq a \leq n | d_a>0, v_a = p\}. \label{eq:V} \end{eqnarray}

A key tool used in our auction is that of $S$-avoid matchings. These are maximal matchings that try to avoid, if at all possible, assigning any items to bidders in some set $S$. Such a matching can be computed by computing a min cost max flow, where there is high cost to direct flow through a vertex of $S$.

In general, the auction prefers to sell items only at the last possible moment (alternately phrased, the highest possible price) at which this item can still be sold while still preserving incentive compatibility. The auction will in fact sell all items (Lemma \ref{lem:sells-all}).

Once a price has been updated, the auction checks to see if it must sell items to value limited bidders. Such bidders will receive no real benefit from the item (their valuation is equal to their payment), but this is important so as to increase the utility of the auctioneer. Our definition of Pareto-optimality includes all bidders and the auctioneer. To check if this is indeed the case, the auction computes a $V$-avoid matching, trying to avoid the bidders in $V$. If this cannot be done, then items are sold to these $V$ bidders. After items are sold to value limited bidders, these bidders effectively disappear by setting their $H_a$ values to $\mathrm{False}$.

The main loop of the mechanism checks whether any items must be sold
to any of the currently active bidders. This is where incentive compatibility comes into play. The auction sells an item to some bidder, $a$, at the lowest price where the remaining bidders total demand is such that an item can be assigned to $a$ without creating a shortage. Again, this makes use of the $\{a\}$-avoid matching, if in the $\{a\}$-avoid matching some item is matched to $a$ then $a$ must be sold that item.

If no items can be sold in this manner, the demand of the bidders is reduced by setting $H_a$ to $\mathrm{False}$, for some active bidder $a$. When neither action can be done, the price increases.

The following lemma shows that all items will in fact be sold.



\begin{algorithm}[p]
\begin{algorithmic}[1]
\Procedure{Combinatorial Auction with Budgets}{$v,b,\{S_i\}$} \newline Implicitly defined $D_a$, $D_a^+$, $d_a$, $U$,  $A$, and  $V$ --- see
Equations (\ref{eq:Diz}) -- (\ref{eq:V}).  \newline $B(\neg \{a\})$ - number of items assigned to agents in $A \setminus \{a\}$ in $\{a\}$-avoid
matching

\State{$p\gets 0$} \While{($A\neq \emptyset$)} \State $\forall a\in A: H_a \gets {\rm True}$  \label{aucline:hitrue}

 \State{Sell($V$)} \label{line:SellV}

\State{$\forall a\in V$: $H_a \gets \rm{False}$ \label{line:VHi}}

 \Repeat \label{aucline:repstart}

 \If{$\exists a | B(\neg\{a\})<m$} {Sell($a$)} \label{line:Selli}

 \Else
\State{For arbitrarily $a\in A$ with $H_a = {\rm True}$ set $H_a \gets \rm{False}$} \label{line:Hifalse} \EndIf

\Until{$\forall a\in A$: $(\neg H_a) ~{\tt and}~ (B(\neg\{a\})\geq m)$} \label{aucline:repend}

\State{Increase $p$ until for some $a\in  A$, $D^+_a$ changes (decreases) } \label{aucline:incprice} \EndWhile \EndProcedure
\end{algorithmic} \caption{Combinatorial Auction with Budgets} \label{alg:combclinch}
\end{algorithm}

\begin{algorithm}[p]
\begin{algorithmic}[1]
\Procedure{$S$-Avoid Matching}{}

Construct interest graph $G$: \begin{itemize} \item Active agents, $A$, on left, capacity constraint of agent $a\in A$ = $d_a$ \item Unsold
items, $U$, on right, capacity constraint $1$.
\item Edge $(a,t)$ from agent $a\in A$ to unsold item $t\in U$ iff $t\in S_a$. \end{itemize}

Return maximal $B$-matching with minimal number of items assigned to agents in $S$, amongst all maximal $B$-matchings.

\EndProcedure
\end{algorithmic} \caption{Computing an avoid matching, can be done via min cost max flow} \label{alg:avoidmatch}
\end{algorithm}

\begin{algorithm}[p]
\begin{algorithmic}[1]
\Procedure{Sell}{$S$}

\Repeat

\State Compute $Y = $ {\sc $S$-Avoid Matching}

\State For arbitrary $(a,t)$ in $Y$, $a \in S$, sell item $t$ to agent $a$.

\Until $B(\neg S) \geq m$
 \EndProcedure
\end{algorithmic} \caption{Selling to a set $S$} \label{alg:sell}
\end{algorithm}

\begin{lemma} \label{lem:sells-all}
If every item appears in $\cup_{i=1}^n S_i$ then the auction will sell all items.
\end{lemma}

\noindent Proof  in Appendix \ref{app:sells-all}.

\section{Pareto-Optimality of the Combinatorial Auction with Budgets}
\label{sec:pareto-opt}

\begin{definition} \label{def:pareto-optimal}
An allocation $(M,P)$ is Pareto-optimal if for no other allocation $(M',P')$ are all players better off, $M'_iv_i - P'_i \ge M_iv_i - P_i$,
including the auctioneer $\sum_i P'(i) \ge \sum_i P_i$, with at least one of the inequalities strict.
\end{definition}

The main goal of this section is to prove the following theorem:

\begin{theorem} \label{thm:main}
The allocation $(M^*,P^*)$ produced by Algorithm \ref{alg:combclinch} is Pareto-optimal. Moreover, the mechanism makes no positive transfers.
\end{theorem}

In Section \ref{subsec:aptppo} we define the notion of trading paths and show the equivalence between allocations with no trading paths and Pareto optimal allocations.
In Appendix \ref{subsec:intrem} we attempt to give some intuition as to why these two are related as well as to why Theorem \ref{thm:main} gives the desirable outcome. In Section \ref{subsec:notradingpaths} we show that the final allocation produced by Algorithm \ref{alg:combclinch} contains no trading paths, thus concluding the proof of Theorem \ref{thm:main}.

\subsection{Alternating paths, Trading paths, and Pareto-optimality}
\label{subsec:aptppo}
\begin{definition}
Consider a path $\pi = (a_1, t_1, a_2, t_2, \ldots, a_{j-1}, t_{j-1}, a_j)$, in a bipartite graph $G$. We say that the path $\pi$ is an
alternating path with respect to $B$-matching $M$ if $(a_i,t_i) \in M$ and $t_i \in S_{i+1}$ for all $1 \leq i<j$. We say that an alternating
path is simple if no agent appears more than once along the path. Note that all alternating paths are of even length (even number of edges).
\end{definition}

\begin{definition}

A path $\pi = (a_1, t_1, a_2, t_2, \ldots, a_{j-1}, t_{j-1}, a_j)$ is called a \emph{trading path} with respect to the allocation
$(M,P)$ if the following hold:

\begin{enumerate}
\item $\pi$ is a simple  alternating path with respect to $M$, (which implies that
agent $a_i$, $i<j$, was allocated item $t_i$ during the course of the auction).
\item The valuation of agent $a_j$, $v_{a_j}$ is strictly greater than the
valuation of agent $a_{1}$, $v_{a_{1}}$.
\item The remaining (unused) budget of agent $a_j$ at the conclusion of the
auction, $b^*_{a_j}$, is $\geq$ the valuation of agent $a_{1}$, $v_{a_{1}}$.
\end{enumerate}
\end{definition}

Intuitively, trading paths, as their name suggests, represent possible trades amongst agents. A trading path allows a trade to take place, where the endpoints of the trading path are better off following the trade, and the interior agents no worse off. (In fact, they can all be made better off by paying a ``commission" of sorts along the path).

We now turn to the following equivalence:

\begin{theorem}
\label{thm:paths} Any allocation $(M,P)$ is Pareto-optimal\footnote{We remark that an analogous (but simpler) claim made in the proceedings
version of the multi unit auction with budget paper \cite{DBLP:conf/focs/DobzinskiLN08} was incorrect but was corrected in \cite{DLNcorrected}.}
{\it if and only if}
\begin{enumerate} \item All items are sold in $(M,P)$, and
\item There are no trading paths in $G$ with respect to $(M,P)$. \end{enumerate}
\end{theorem}

\noindent Proof in Appendix \ref{app:paths}.

\subsection{No Trading Paths in $(M^*,P^*)$}
\label{subsec:notradingpaths}

To conclude the proof of Theorem \ref{thm:main} we now prove that there are no trading paths in the final allocation $(M^*,P^*)$ generated by the mechanism given in Algorithm \ref{alg:combclinch}.

We know
from Lemma \ref{lem:sells-all} that $M^*$ matches all items.

Consider the set of all trading paths $\Pi$ in the final allocation $M^*$.
\begin{definition} \label{def:suppi}
We define the following for every $\pi \in \Pi$: \begin{itemize}
\item Let $Y^{\pi}$ be the $S$-avoid matching used the first time some item $t$ is sold to some agent $a$ where $(a,t)$ is an edge along $\pi$.
$Y^{\pi}$ is either a $V$-avoid matching (line \ref{line:SellV} of Algorithm \ref{alg:combclinch}) or an $a$-avoid matching for some agent-item
edge $(a,t)$ along $\pi$ (line \ref{line:Selli} of Algorithm \ref{alg:combclinch}).
\item If $Y^{\pi}$ is a $V$-avoid matching, let $V^{\pi}$ be this set of value limited agents.
\item If $Y^{\pi}$ is an $a$-avoid matching, let $a^{\pi}$ be this agent.
\item Let $F^{\pi}\subset M^*$ be the set of edges $(a,t)$ such that item $t$ was sold to agent $a$ at or subsequent to the first time
that some item $t'$ was sold to some agent $a'$ for some edge $(a',t')\in \pi$ ($(a',t')$ is itself in $F^{\pi}$).
\item Let $m^{\pi}$ be the number of unsold items just before the first time some edge along $\pi$ was sold. {\sl I.e.},
$m^{\pi}$ is equal to the number of items matched in $F^{\pi}$.
\item Let $p^{\pi}$ be the price at which item[s] were sold from $Y^{\pi}$.
\item Let $b^{\pi}_a$ be the remaining budget for agent $a$ before any items are sold in Sell($V^{\pi}$) or Sell($a^{\pi}$).
\end{itemize}
\end{definition}
We partition $\Pi$ into two classes of trading paths:
\begin{enumerate} \item $\Pi_V$ is the set of trading paths such that $\pi \in \Pi_V$ iff $Y^{\pi}$ is some $V^{\pi}$-avoid matching
used in Sell($V^{\pi}$) (line \ref{line:SellV} of Algorithm \ref{alg:combclinch}).
\item $\Pi_{\neg{V}}$ is the set set of trading paths such that $\pi \in \Pi_{\neg{V}}$ iff
$Y^{\pi}$ is some $a^{\pi}$-avoid matching used in Sell($a^{\pi}$) (line \ref{line:Selli} of Algorithm \ref{alg:combclinch}). \end{enumerate}

\begin{lemma} \label{lem:SellV} $\Pi_V=\emptyset$.
\end{lemma}
\begin{proof}

We need the following Claim:

\begin{claim} Given a trading path $\pi = (a_1,t_2, \ldots, a_{j-1},t_{j-1},a_j)\in \Pi_V$, and let $(a_i,t_i)$
be the last edge belonging to $Y^{\pi}$ along $\pi$. Then the suffix of $\pi$ starting at $a_i$, $(a_i, t_i, \ldots, a_j)$, is itself a
trading path.
\end{claim}
\begin{proof}
This trivially follows as the valuation of $a_i$ is equal to current price when Sell($V^{\pi}$) was done ($p^{\pi}$), and the valuation of $a_1$
is $\geq p^{\pi}$ as edge $(a_1,t_1)$ was unsold prior to this Sell($V^{\pi}$) and does belong to the final $F^{\pi}$.
\end{proof}

From the Claim above we may assume, without loss of generality, that if $\Pi_V \neq \emptyset$ then $\exists \pi \in \Pi_V$ such that the first
edge along $\pi$ was also the first edge sold amongst all edges of $\pi$, furthermore, all subsequent edges do not belong to $Y^{\pi}$.

 As agents $a\in V^{\pi}$ will not be sold any further items after this Sell($V^{\pi}$), the items assigned to
$a_1$ in $Y^{\pi}$ are the same items assigned to $a_1$ in $F^{\pi}$.

We seek a contradiction to the assumption that $Y^{\pi}$ was a $V^{\pi}$-avoid matching. Note that the matching $F^{\pi}$ is a $V^{\pi}$-avoid matching by itself, because exactly the items assigned to $V$-type agents in $Y^{\pi}$ are sold. We now show how to construct from $F^{\pi}$ another matching that assigns less items to $V$-type agents.

We show that the number of items assigned to agent $a_1$ in $F^{\pi}$ (which is the same as in $Y^{\pi}$) can be reduced by one by giving agent $a_{k+1}$ item $t_k$ for $k=1, \ldots, j-1$. This is also a full matching but it remains to show that this does not exceed the capacity constraints for agent $a_j$, $d_{a_j}$.

As $H_a = \rm{True}$ for all $a\in A$ when Sell($V^{\pi}$) is done, this means that $d_{a_j}=D_{a_j}$. Agent $a_j$ has remaining budget $\geq v_1$ at the conclusion of the auction, and all items assigned to agent $a_j$ in $F^{\pi}$ are at price $\geq p^{\pi} = v_1$. This implies that at the time of Sell($V^{\pi}$) we have $D_{a_j} > $ the number of items assigned to $a_j$ in $F^{\pi}$. Thus, we can increase the number of items allocated to $a_j$ by one without exceeding the demand constraint $d_{a_j}=D_{a_j}$.

Now, note that $a_j$ is not $V$-type agent, so the new matching constructed assigns less items to $V$ type agents then the matching $F^{\pi}$. Hence, $F^{\pi}$ is not an $V^{\pi}$-avoid matching, and in turn neither $Y^{\pi}$ is $V^{\pi}$-avoid matching.
\end{proof}

We've shown that $\Pi_V=\emptyset$. It remains to show that $\Pi_{\neg{V}}=\emptyset$.

Assume $\Pi_{\neg{V}} \neq \emptyset$. Order $\pi \in \Pi_{\neg{V}}$ by the first time at which some edge along $\pi$ was sold. We know that
this occurs within some Sell($a^{\pi}$) for some $a^{\pi}$ and that $a^{\pi} \notin V$. Let us define $\pi = (a_1, t_1, a_2, t_2, \ldots, a_{j-1},
t_{j-1}, a_j)$ be the last path in this order, and let $e=(a^{\pi},t^{\pi})=(a_i,t_i)$.

Recall that $Y^{\pi}$ is the $a^{\pi}$-avoid matching used when item $t^{\pi}$ was sold to agent $a^{\pi}$. Also, $F^{\pi}\subset M^*$ is the
set of edges added to $M^*$ in the course of the auction from this point on (including the current Sell($a_i$)).

\begin{lemma} \label{lem:x} Let $\pi$, $a^{\pi} = a_i$, $t^{\pi}=t_i$, be as above, we argue that when $Y^{\pi}$ was computed as an $a^{\pi}$-avoid matching
there was another full matching $X$ with the following properties:
\begin{enumerate}
\item The suffix of $\pi$ from $a_i$ to $a_j$: $$\pi[a_i,\ldots,a_j] = (a_i,t_i,a_{i+1},t_{i+1},\ldots, a_{j-1},t_{j-1}, a_j),$$ is an alternating path
with respect to $X$. (\sl I.e., edges $(a_k,t_k)$, $i \leq k \leq j-1$, belong to $X$).
\item The number of items assigned to $a_i$ in $X$ is equal to the number of items assigned to $a_i$ in $Y^{\pi}$.
\item The number of items assigned to $a_j$ in $X$ is equal to the number of items assigned to $a_j$ in $F^{\pi}$.
\end{enumerate}
\end{lemma}

\begin{proof}
Consider the final matching $F^{\pi}$. Note that $F^{\pi}(a_i) \geq Y^{\pi}(a_i)$, because otherwise if $F^{\pi}(a_i) < Y^{\pi}(a_i)$ then
$F^{\pi}(a_i)$ would have fewer items assigned to $a_i$ than the $a_i$-avoid matching $Y^{\pi}$, a contradiction.

 If $F^{\pi}(a_i) =
Y^{\pi}(a_i)$ then choose $X=F^{\pi}$ and conditions 1 -- 3 above hold trivially.

Thus, we are left with the case where $F^{\pi}(a_i) > Y^{\pi}(a_i)$. Consider the symmetric difference $F^{\pi} \oplus Y^{\pi}$. By Lemma
\ref{lem:paths-and-cycles} the edges of $F^{\pi} \oplus Y^{\pi}$ can be covered by alternating paths with respect to $F^{\pi}$. There must be
 $\delta = F^{\pi}(a_i) - Y^{\pi}(a_i)$ such paths starting at agent $a_i$ (as agent $a_i$ has $\delta$ more items assigned
 in $F^{\pi}$ than in $Y^{\pi}$). Take one of these paths $\tau = (a_i=g_1, s_1, g_2, s_2, \ldots, g_{\ell})$, $g_k$'s are agents, $s_k$'s are
 items, $(g_k,s_k)$ belongs to $F^{\pi}$, $(s_k,g_{k+1})$ belongs to $Y^{\pi}$.

 We now argue that $\tau$ and $\pi[a_i,\ldots,a_j]$ are vertex
 disjoint besides the first agent $a_i$. To reach a contradiction, assume that there is another common vertex $u$ along $\tau$ and
 along $\pi[a_i,\ldots,a_j]$,
 $u\neq a_i$. Choose $u$ to be the first such vertex along $\tau$.

 We consider two possibilities:
 \begin{enumerate}
 \item
 $u$ is an item. Consider $$\pi[a_i,\ldots,a_j]=(a_i, t_i, a_{i+1}, t_{i+1}, \ldots, a_{j-1}, t_{j-1}, a_j),$$ and let  $u=s_k=t_{k'}$
 for some $k,k'$. Then both $(g_k,s_k = t_{k'} =u )$ and $(a_{k'},s_k=t_{k'}=u)$ belong to $F^{\pi}$.
 This implies either that  item $u$ is assigned to two different agents in $F^{\pi}$ or that $a_{k'}=g_k$
 in contradiction to our choice of $u$ as the first common vertex along $\tau$.
 \item $u$ is an agent. For some $i < k \leq j$, $1 < k'\leq \ell$, $u=g_k=a_{k'}$.
 Let $\pi'$ be the concatenation of the prefix of $\pi$ up to $a_i$, followed by the prefix of $\tau$ up to
$g_k$ and then followed by the suffix of $\pi$ from $g_k = a_{k'}$ to the end:
 $$\pi' = (a_1,t_1,\ldots, a_i=g_1, s_1, g_2, \ldots, g_k = a_{k'}, t_{k'}, a_{k'+1}, \ldots, a_j).$$ This path is a trading path in $F^{\pi}$, and none of
 the edges along this path were sold before the edge $(a_i,t_i)$, in contradiction to the assumption that $\pi$ had it's first sold edge sold
 last amongst all trading paths.
 \end{enumerate}

Therefore, $\tau$ and $\pi[a_i,\ldots,a_j]$ only have $a_i$ in common. By Lemma \ref{lem:paths-and-cycles} the different paths $\tau$ starting
from $a_i$ in $Y^{\pi}\oplus F^{\pi}$ are edge disjoint. For any such $\tau = (a_i=g_1, s_1, g_2, s_2, \ldots, g_{\ell})$, agent $g_k$ holds
item $s_k$ in $F^{\pi}$, $1 \leq k \leq \ell-1$, and agent $g_{k+1}$ holds item $s_k$ in $Y^{\pi}$, $1 \leq k \leq \ell-1$. Therefore, we can
move item $s_k$ from agent $g_k$ to agent $g_{k+1}$, $1 \leq k \leq j-1$, without violating the demand of agent $g_{\ell}$ because $s_{\ell-1}$
was assigned to $g_{\ell}$ in $Y^{\pi}$. As we can do so for all such paths $\tau$ we obtain a new full matching $X$ where the number of items
assigned to agent $a_i$ is the same as the number of items assigned to agent $a_i$ in $Y^{\pi}$.

Note that, other than $a_i$, none of the agents along the path $\pi[a_i,\ldots,a_j]$ appears on any of these $\tau$ and therefore their
assignment in $X$ remains unchanged from their assignment in $F^{\pi}$.

\end{proof}

\begin{corollary} $\Pi_{\neg{V}} = \emptyset$.

\end{corollary}
\begin{proof}
Assume $\pi\in \Pi_{\neg{V}} \neq \emptyset$ and let $a^{\pi} = a_i$, $t^{\pi} = t_i$,
  we now seek to derive a contradiction as follows:
  \begin{itemize} \item When $Y^{\pi}$ was computed there was also an an alternate full matching $Y'$ with fewer
  items assigned to agent $a_i$, contradicting the assumption that $Y^{\pi}$ is an $a_i$ avoid matching. Or,
  \item We show that the remaining budget of agent $a_j$ at the end of the auction, $b^*_{a_j}$, has $b^*_{a_j} < v_1$, contradicting the assumption that $\pi$ is
  a trading path.
  \end{itemize}

Let $X$ be a matching as in Lemma \ref{lem:x} and $F^{\pi}$ be as defined in Definition \ref{def:suppi}. Also, let $X(a)$, $F^{\pi}(a)$, be the
number of items assigned to agent $a$ in full matchings $X$, $F^{\pi}$, respectively.

We consider the following cases regarding $d_{a_j}$ when $Y^{\pi}$, the $a_i$-avoid matching, was computed:
\begin{enumerate}
  \item $d_{a_j}>X(a_j)$: then, like in Lemma \ref{lem:SellV}, we can decrease the number of items sold to $a_i$ by assigning
  item $t_k$ to agent $a_{k+1}$ for $k=i, \ldots, j-1$, without exceeding the $d_{a_j}$ demand constraint.
  \item $d_{a_j} = X(a_j)$, by subcase analysis we show that $b^{\pi}_{a_j} \leq (X(a_j)+1)p^{\pi}$:
  \begin{enumerate}
    \item $D_{a_j} = D^+_{a_j}$: Observe that $X(a_j) < m$, the current number of unsold items. This follows because $X(a_i) = Y^{\pi}(a_i) \geq 1$ by assumption
    that $t_i$ was assigned to $a_i$ in $Y^{\pi}$. This means that $d_{a_j} = X(a_j) < m$ so  \begin{eqnarray*} X(a_j) &=& d_{a_j} =
    \left\lfloor b^{\pi}_{a_j}/p^{\pi} \right\rfloor > b^{\pi}_{a_j}/p^{\pi} -1 \\
    \Rightarrow b^{\pi}_{a_j} &<& (X(a_j)+1)p^{\pi}. \end{eqnarray*}
    \item $D_{a_j} \neq D^+_{a_j}$: Observe that $a_j \notin V$ as $v_{a_j} > v_{a_i}$ and $a_i \notin V$.
    As $a_j\notin V$, the only reason that $D_{a_j} \neq D_{a_j}^+$ is because the remaining budget of agent $a_j$, $b^{\pi}_{a_j}$, is
    an integer multiple of the current price $p^{\pi}$. Then, $D_{a_j}^+ = D_{a_j}-1$ and $D_{a_j} = \lfloor b^{\pi}_{a_j}/p^{\pi} \rfloor =
    b^{\pi}_{a_j}/p^{\pi}$, it follows that  \begin{eqnarray*} X(a_j) &=& d_{a_j} \geq D_{a_j}^+ = D_{a_j} -1 = b^{\pi}_{a_j}/p^{\pi} -1 \\
    \Rightarrow b^{\pi}_{a_j} &\leq& (X(a_j)+1)p^{\pi}. \end{eqnarray*}
  \end{enumerate}
  Note that the current price $p^{\pi}<v_{a_i}$ because we assume that $a_i$ was sold $t_i$ as a result of Sell($a_i$) and not Sell($V$). It is
  also true that $v_{a_i} \leq v_{a_1}$ as $(a_i,t_i)$ was the first edge that was sold along $\pi$.
  By condition 3 of Lemma \ref{lem:x} we can deduce that
$$b^{\pi}_{a_j} \leq (X(a_j) + 1) p^{\pi} = (F^{\pi}(a_j)+1)p^{\pi}.$$
Agent $a_j$ is sold exactly $F^{\pi}(a_j)$ items at a price not lower that $p^{\pi}$, to at the end of the auction the remaining budget for
agent $a_j$, $b^*_{a_j}$, is $\leq p^{\pi}$.  This contradicts the assumption that $\pi$ is a trading path since $$b^*_{a_j} \leq p^{\pi} <
v_{a_i} \leq v_{a_1}.$$
\end{enumerate}

\end{proof}

\section{Mapping the Frontier} \label{sec:impossible}

In this paper we gave a mechanism that is incentive compatible with respect to valuation, and produces a Pareto-optimal allocation, but with various annoying restrictions and assumptions:
\begin{itemize} 
\item we assume public budgets;
\item we assume public sets of interest;
\item moreover, agents are restricted to have a step function valuation for items, if the item is in $S_i$ then it's valuation is $v_i$, otherwise zero.
\end{itemize}

This poses the question: can we remove these annoying assumptions/restrictions? Just how far can we go?

As for private budgets, it was shown by \cite{DBLP:conf/focs/DobzinskiLN08} that even for the multi unit case, one cannot achieve incentive compatibility with respect to valuation along with bidder rationality, auctioneer rationality, and obtain a Pareto-optimal allocation.

We argue that even if one assumes public budgets, the other restrictions are also necessary. This is summarized in the following theorems:

\begin{theorem} \label{thm:noprivateS}
There is no truthful, bidder rational, auctioneer rational and
Pareto-optimal auction with
public budgets, $b_a$, private valuations, $v_a$, and private sets of
interest, $S_a$.
\end{theorem}

Proof in Appendix \ref{app:noprivateS}.

\begin{corollary} \label{thm:noprivatev}
There is no truthful, bidder rational, auctioneer rational and
Pareto-optimal auction with
public budgets, $b_a$, and private item-dependent valuations $v_{at}$.
\end{corollary}
\begin{proof}
  This follows immediately from Theorem \ref{thm:noprivateS}. Consider the case where the private valuations $v_{at}$ are zero for any $t\notin S_a$, and $v_a$ for $t\in S_a$.
\end{proof}

\bibliographystyle{plain}
\bibliography{clinch}

\appendix
\newpage

\begin{table}[h]
\begin{center}
\begin{tabular}{|@{$\quad$}l@{$\quad$}|p{11.5cm}|}
\hline & \\
Notation & Explaination \\ & \\ \hline
$n$ & Number of agents \\
$m$ & Current number of items  \\
$S_a$  & Items agent $a$ is interested in \\
$v\in \Re^m$ & $v_a>0$ is the valuation of agent $a$ for the items in $S_{a}$ \\
$b \in \Re^m$ & $b_a$ is the current budget for agent $a$ \\
$p \in \Re^+$ & The  current price \\
$A$ & Current  active agents ($d_a >0$)\\
$V$ & Current  value limited agent ($d_a>0, v_a = p$) \\
$U$ & Current set of unsold items\\
$D_a$ & $\left\{\begin{array}{ll}
{\tt min} \{m,\lfloor b_i/p \rfloor\} & {\rm if\ } p\leq v_i  \\
0 & {\rm if\ } p>v_i
\end{array} \right.$  \\
$D^+_a$ & $D_a$ at infinitesimally higher price than $p$\\
$d_a$ & $D_a$ if $H_a=\rm{True}$, $D^+_a$ otherwise\\
$H_a$ & Boolean value, if true $d_a=D_a$, OW $d_a=D_a^+$ \\
$(M^*,P^*)$ & The matching and payments resulting from the auction \\
$M_i$ & The number of items sold to agent $i$ in matching $M$ \\
$P_i$ & The total payment by agent $i$ given payment vector $P\in \Re^n$ \\
$\Pi$ & The set of all trading paths in $M^*$ \\
$\pi \in \Pi$ & A trading path $(a_1, t_1, \ldots, a_{j-1}, t_{j-1}, a_j)$ \\
$\pi[a_i,\ldots,a_j]$ & A suffix of $\pi$: $(a_i, t_i, \ldots, a_j)$ \\
$V^{\pi}$ & First time any edge was sold from  $\pi$ was during Sell($V^{\pi}$) \\
$a^{\pi}$ & First time any edge was sold from  $\pi$ was during Sell($a^{\pi}$) \\
$Y^{\pi}$ & Either $V^{\pi}$-avoid matching or $a^{\pi}$-avoid matching \\
$\Pi_V$ & First time any edge was sold from  $\pi\in \Pi_V$ was during Sell($V^{\pi}$) \\
$\Pi_{\neg V}$ & First time any edge was sold from  $\pi\in \Pi_{\neg V}$ was during Sell($a^{\pi}$) \\
$b_a^{\pi}$ & Budget of agent $a$ before 1st time any edge sold from $\pi$ \\
$b_a^*$ & Remaining budget of agent $a$ at end of auction \\
$B(\neg S)$ & $\#$ items assigned to agents in $A\setminus S$ in $S$-avoid matching
 \\ \hline
\end{tabular}
\end{center}
\caption{Notation Used} \label{tab:notation}
\end{table}

\appendix

\section{Proof of Lemma \ref{lem:sells-all}}
\label{app:sells-all}

\begin{proof}
We prove that throughout the auction, there is always a matching that can sell all remaining items at the current price without exceeding the
budget of any agent. As prices only increase, eventually all items must be sold. The lines below refer to Algorithm \ref{alg:combclinch} unless
stated otherwise.

Initially, all items can be sold at price zero. The $d_a$ capacity constraints are all equal to $m$.

Furthermore, we argue that is is always true that all unsold items can be sold to active agents at the current price without violating the
capacity constraints. We prove this invariant by case analysis of the following events:
\begin{itemize} \item Increase in price followed by setting the $H_a$ variables to $\rm{True}$: The repeat loop in lines \ref{aucline:repstart}
-- \ref{aucline:repend} ends with $H_a\gets\rm{False}$ and $B(\neg\{a\})\geq m$ for all agents $a$. Thus, when the condition in line
\ref{aucline:repend} is met, all the $d_a$'s are set to $D_a^+$.

Any increment in price in line \ref{aucline:incprice} will set $D_a$ equal to the previous $D_a^+$ and the subsequent assignment of
$H_a\gets\rm{true}$ (line \ref{aucline:hitrue}) means that the new $d_a's$ are equal to the old ones. Thus, any matching valid at the old price
is valid at the new price.
\item The Sell($V$) operation (line \ref{line:SellV} of Algorithm \ref{alg:combclinch}, Algorithm \ref{alg:sell}) sells items to
agents in $V$ only if all other unsold items can be matched to agents not in $V$.
\item Setting $H_a\gets \rm{False}$ for $a\in V$ (line \ref{line:VHi}) sets $d_a=0$ for $a\in V$ and this is OK because nothing will be sold to
$a\in V$ at any higher price.
\item The Sell($a$) operation (line \ref{line:Selli} of Algorithm \ref{alg:combclinch}, Algorithm \ref{alg:sell})
sells items to agent $a$ only if all other unsold items can be matched to other agents.
\item Setting $H_a \gets \rm{False}$ (line \ref{line:Hifalse}) is done only if $B(\neg\{a\})\geq m$, {\sl i.e.},
all unsold items can be matched to the other agents (not
including $a$).
\end{itemize}

Thus, the mechanism will sell all items.
\end{proof}

\section{Proof of Theorem \ref{thm:paths}}
\label{app:paths}

\begin{proof} Let $Q$ be the predicate that $(M,P)$ is Pareto-optimal, $R_1$ be the predicate that all items are sold in $(M,P)$, and $R_2$ the
predicate that there are no trading paths in $G$ with respect to $(M,P)$. We seek to show that $Q \Leftrightarrow R_1 \cap R_2$.

$Q \Rightarrow (R_1\cap R_2)$: to prove this we show that $(\neg R_1 \cup \neg R_2) \Rightarrow \neg Q$.

If both $R_1$ and $R_2$ are true then this becomes $\mathrm{False} \Rightarrow Q$ which is trivially true.

If the allocation
$(M,P)$ does not assign all items ($\neg R_1$) then it is clearly not Pareto-optimal ($\neg Q$). We can get a better allocation by assigning all
unsold items to any agent $i$ with such items in $S_i$. This increases the utility of agent $i$.

If $\neg R_2$ then there exists a trading path in $G$ with respect to $(M,P)$, let this path be $\pi = (a_1, t_1, a_2, t_2, \ldots, a_{j-1}, t_{j-1}, a_j)$, as $v_{a_j} > v_{a_1}$ and $b^*_{a_j} \geq v_{a_1}$ then we can decrease the payment of agent $a_1$ by $v_{a_1}$, increase the payment of agent $a_j$ by the same $v_{a_1}$, and move item $t_i$ from agent $a_i$ to agent $a_{i+1}$ for all $i=1,\ldots, j-1$. In this case, the utility of agents $a_1, a_2, \ldots, a_{j-1}$ is unchanged, the utility of agent $a_j$ increases by $v_{a_j}-v_{a_i} > 0$, and the utility of the auctioneer is unchanged. The sum of payments by the agents is likewise unchanged. This contradicts the assumption that $(M,P)$ is Pareto optimal.

We now seek to prove that $(R_1 \cap R_2) \Rightarrow Q$. We note above that if not all items are allocated ($\neg R_1$) then the allocation is
not Pareto-optimal ($\neg Q$), thus $Q \Rightarrow R_1$ and (trivially) $Q \Rightarrow Q \cap R_1$ (Pareto optimality implies all items
allocated). Thus, $(R_1 \cap R_2) \Rightarrow Q \Rightarrow Q \cap R_1$. If $R_1$ is false this predicate becomes $\mathrm{False} \Rightarrow
\textrm{False}$, thus we remain with the case where all items are allocated.

Assume $\neg Q$, {\sl i.e.}, assume that $(M,P)$ is not Pareto-optimal --- then there must be some other allocation $(M',P')$ that is no worse
for all players (including the auctioneer) and strictly better for at least one player. We can assume that $(M',P')$ assigns all items as well,
as otherwise we can take an even better allocation that would assign all items.

By Lemma~\ref{lem:paths-and-cycles} (see below) we know that $M$ and $M'$ are related by a set of simple paths and cycles. On a path, the first agent gives up one item, whereas the last agent receives one item more, after items are exchanged along the path. Cycles represents giving up one item in
return for another by passing items around along it. Cycles don't change the number of items assigned to the bidders along the cycles so we will ignore them.
$x_1,\ldots,x_z$ and $y_1,\ldots, y_z$ denote the start and end agents along these $z$ alternating paths.
Note that the same agent may appear multiple times amongst $x_i$'s or multiple times amongst $y_i$'s, but cannot appear both as an $x_i$ and as a $y_i$ (we can concatenate two such paths into one). Such an alternating path
represents a shuffle of items between agents where agent $x_j$ looses an item whereas agent $y_j$ gains an item  when  moving from $M$ to $M'$.
In general, these two items may be entirely different.

Since there are no trading paths with respect to $(M,P)$, it must be the case that for every one of these $z$ alternating paths either
\begin{enumerate} \item[$\alpha$.]   $v_{y_j} \leq v_{x_j}$ holds. Define $I = \{j | v_{y_j} \leq v_{x_j}\}$.
\item[$\beta$.] $b^*_{y_j} < v_{x_j}$ holds (where $b^*_{y_j}$ is the budget left over for agent $y_j$ at the end of the mechanism).
 Define $J =\{ j | b^*_{y_j} < v_{x_j}\}$.
 \end{enumerate}

Now, no bidder is worse off in $(M',P')$ (in comparison to $(M,P)$), and the auctioneer is no worse off, and, by assumption, either/or \begin{enumerate} \item[A.] Some bidder is strictly better off. Or, \item[B.] The auctioneer is strictly better off.
\end{enumerate}

First, we rule out case B above:
 Consider the process of changing $(M,P)$ into $(M',P')$ as a two stage process: at first, the agents $x_1, \ldots, x_z$ give up items. During this first stage, the payments made by agents $x_1, \ldots, x_m$ must decrease (in sum) by at least $Z^-=\sum_{i=1}^z{v_{x_i}}$. The 2nd stages is that agents ${y_1,\ldots,y_z}$ receive their extra items. In the 2nd stage, the maximum extra payment that can be received from agents ${y_1,\ldots,y_z}$ is no more than \begin{equation} Z^+=\sum_{j\in I} v_{y_j} + \sum_{j\in J} b^*_{y_j} \leq \sum_{j\in I} v_{x_j} + \sum_{j\in J} v_{x_j} = Z^-,\label{eq:zplezm}\end{equation} by definition of sets $I$ and $J$ above. Thus, the total increase in revenue to the auctioneer is $Z^+ - Z^- \leq 0$. This rules out Case B above (auctioneer strictly better off). Moreover, as the auctioneer cannot be worse off, $Z^+ = Z^-$ and from Equation (\ref{eq:zplezm}) we conclude that \begin{equation}\sum_{j\in I} v_{y_j} + \sum_{j\in J} b^*_{y_j} = \sum_{j\in I} v_{x_j} + \sum_{j\in J} v_{x_j}. \label{eq:equal}\end{equation}

From $\alpha$ above, we have that $v_{y_j} \leq v_{x_j}$ for $j\in I$, from $\beta$ be have that $b^*_{y_j} < v_{x_j}$ for $j\in J$.
Thus, if $J\neq \emptyset$ then the lefthand side of Equation (\ref{eq:equal}) is strictly less than the righthand side, a contradiction.

Therefore, case A must hold and it must be that $J=\emptyset$, we will conclude the proof of the theorem by showing that these two are inconsistent. So, we have that
\begin{eqnarray*} M'_a  v_{a} - P'_{a} &=& M_a v_a - P_{a} \quad \mbox{\rm for agents $a$ whose utility is unchanged}\\
M'_{\hat{a}}  v_{\hat{a}} - P'_{\hat{a}} &>& M_{\hat{a}} v_{\hat{a}} - P_{\hat{a}} \quad \mbox{\rm for some agent $\hat{a}$}\\
\sum_a P'_a &=& \sum_a P_a.
\end{eqnarray*}

We can now derive that
\begin{eqnarray}
   \sum_a M'_a  v_{a}   &>& \sum_a M_a v_a  - \left(\sum_a P'_a - \sum_a P_a\right) \nonumber\\
  &=& \sum_a M_a v_a. \nonumber \\
    \Rightarrow \quad \sum_a (M'_a - M_a) v_a &>& 0. \label{eq:sumdiffm}\end{eqnarray}

  Now, whenever $a = x_j$ we decrease $M'_a - M_a$ by one, whenever $a= y_j$ we increase $M'_a - M_a$ by one.
  Thus, rewriting Equation (\ref{eq:sumdiffm}) we get that
  \begin{eqnarray} \sum_a  (|\{j| a=y_j\}| - |\{j|a=x_j\}|) v_a &>& 0 \nonumber \\
  \Rightarrow \sum_{j=1}^z v_{y_j} - \sum_{j=1}^z v_{x_j} &>& 0 \nonumber\\
  \Rightarrow \sum_{j=1}^z v_{y_j} &>& \sum_{j=1}^z v_{x_j}. \label{eq:sumj1z}
   \end{eqnarray}

But, Equation (\ref{eq:sumj1z}) is inconsistent with Equation (\ref{eq:equal}) as $J=\emptyset$ implies that $I=\{1, \ldots, z\}$.

\end{proof}

The following technical lemma was required in the proof of Theorem \ref{thm:paths} above:
\begin{lemma}
\label{lem:paths-and-cycles} Let $M$ and $M'$ be two $B$-matchings that allocate all items, then, the symmetric difference between these two
matchings, $M \oplus M'$, can be decomposed into a set of simple alternating paths (with respect to $M$) and alternating cycles (also with
respect to $M$) that are edge disjoint. Moreover, there are no two simple alternating paths such that one ends and the other begins at the same
agent.
\end{lemma}
\begin{proof}
Intuitively, the set $M \oplus M'$ relates $M$ to $M'$ and shows how to change one matching into another. To prove the lemma, direct edges in
$M$ from agents to items and edges in $M'$ from items to agents. Denote the resulting graph as $\vec{G}$. Any directed graph (and $\vec{G}$ in
particular) can be decomposed into a set of simple paths and cycles, such that no two simple paths start and end in the same vertex, {\sl i.e.},
maximal length simple paths.

To prove that such paths cannot start or end at an item, recall that both $M$ and $M'$ allocate all items. Thus, every item is adjacent to one
edge in $M$ and one edge in $M'$, so in $M\oplus M'$ it is adjacent to either zero or to 2 edges. Should we assume that some path starts at an
item, this  contradicts our assumption of maximal paths in $\vec{G}$.  A similar argument shows that no path can end at an item. Therefore, all
paths start and end at an agent. The maximality of the paths in $\vec{G}$ also shows that there are no two paths such that one ends and the
other begins at the same agent.

Along any such path or cycle, there can be no two consecutive edges from $M$ and there can be no two consecutive edges from $M'$. Also, for all
edges in $M\oplus M'$ between an agent $i$ and an item $j$, it must be that $j\in S_i$. Thus all maximal paths and all cycles covering $\vec{G}$
are alternating paths with respect to $M$. We also remark that should we reverse the direction of the paths and cycles then they will be
alternating paths with respect to $M'$.

\end{proof}

\section{Discussion and Remarks}
\label{subsec:intrem}

We hope that the following remarks may prove helpful:
\begin{enumerate}
    \item In the definition of Pareto-optimality (Definition \ref{def:pareto-optimal}), one allows any alternative allocation and pricing. If (for example) we were to redefine Pareto optimality, defining ``Pareto-optimality" by appending to the sentence fragment ``for no other allocation $(M',P')$" the suffix ``{\sl such that $P'_i\geq 0$ for all $i$}". Then, ``Pareto-optimal" assignments could in fact contain trading paths. Such trades would be ``illegal"  because they would violate the no positive transfers condition ($P'_i \geq 0$).
  \item Pareto-optimality as given in Definition \ref{def:pareto-optimal} is a more desirable social goal than ``Pareto-optimal" invented above. If we only insisted on a ``Pareto-optimal" assignment, then we could get very bad assignments. Later, subsequent to the auction, the bidders could trade amongst themselves and improve their lot.
      \item However, it may also be desirable that no agent actually get paid from the mechanism. Thus, it may be desirable that the actual allocation produced by the action have no positive transfers ($P_i \geq 0$ for all $i$), yet at the same time be Pareto-optimal in the strong sense of Definition \ref{def:pareto-optimal}: after the allocation is presented, no agents will desire to trade amongst themselves. This is the claim of Theorem \ref{thm:main}.
  \end{enumerate}

\section{Proof of Theorem \ref{thm:noprivateS}}
\label{app:noprivateS}

For the proof of  Theorem \ref{thm:noprivateS}  (up to by not including Corollary \ref{thm:noprivateS}) we assume the step function valuations (as done throughout this paper).

Recall the uniqueness result of~\cite{DBLP:conf/focs/DobzinskiLN08}:

\begin{theorem}[Theorem 5.1 of \cite{DBLP:conf/focs/DobzinskiLN08}]
\label{thm:dnl-uniqueness}
Let $A$ be a truthful, bidder-rational, auctioneer rational, and Pareto-optimal multi unit auction (identical items) with $2$ players with known (public) budgets $b_1$, $b_2$ that are generic\footnote{Not all pairs of values are generic, but for our purposes assume that this holds for every such pair.} then if $v_1 \neq v_2$ the allocation produced by $A$ is identical to that produced by the Dynamic clinching auction of \cite{DBLP:conf/focs/DobzinskiLN08} (and, in particular, with our auction when applied to these inputs).
\end{theorem}

For all the details of the proof please see~\cite{DLNcorrected}, as the original publication~\cite{DBLP:conf/focs/DobzinskiLN08} includes only a sketch.


\subsection{Public budgets $b_i$, Private valuations $v_i$, and Private sets of interest $S_i$} \label{sec:nopublicsi}
We now show that there is no incentive compatible, Pareto-optimal, bidder rational, and auctioneer rational mechanism when the budgets are public, and the agent valuation and set of interest is private.

We say that an agent wins an item if the item is assigned to the agent.

Consider two agents, $1$ and $2$, and two items $t_1$, $t_2$. Let $S_{1} = \{t_1\}$ and $S_{2} = \{ t_1, t_2\}$. We now prove the following:
\begin{lemma}
\label{lem:lower-1}
Consider any incentive compatible, Pareto-optimal, bidder rational and auctioneer rational combinatorial auction that produces an allocation $(M,P)$:  if agent $2$ wins both items
than the payment $P_1$ by agent $1$ is zero.
\end{lemma}
\begin{proof}
First, consider the case when $v_1=0$. Then any incentive compatibility and Pareto-optimality auction has to assign both items
to agent $2$. If any of the items were to be left unassigned, or would be assigned to agent $1$, we could assign it to agent $2$, without changing any payment. This does not change the utility of agent $1$, nor the utility of the auctioneer, but would strictly increase  the utility of agent $2$.

Because of incentive compatibility, agent $2$ pays $P_2 = 0$. Otherwise, agent 2 could reduce his reported valuation and attain the item at a lower price. If follows from bidder rationality that $P_1 \le 0$ (we have not ruled out positive transfers yet). However,it follows from auctioneer rationality that agent one must pay zero, as $-P_1 \le P_2 = 0$.

Now, consider the case when both agents have nonzero valuations. Then for every instance in which agent $1$ gets no items it must be that $P_1 = 0$. By IC his payment cannot depend on his valuation, and when agent 1 reported a valuation of zero then $P_1$ was zero.
\end{proof}

\begin{lemma}
\label{lem:lower-2}
Consider any incentive compatible, Pareto-optimal, bidder rational and auctioneer rational combinatorial auction that produces an allocation $(M,P)$: if agent $2$ does not win item $t_1$
then $P_2=0$.
\end{lemma}
\begin{proof}
First consider the case when $v_2=0$, and $v_1>0$. As in previous proof, any incentive compatibility and Pareto-optimality auction has to assign item $t_1$ to agent $1$. It follows from incentive compatibility that agent $1$ pays $P_1 = 0$, whereas it follows from bidder rationality and auctioneer rationality that $P_2=0$.

Now, consider the case when both agents have nonzero valuations. On every input when agent $2$ is not assigned item $t_1$, it must be that $P_2=0$, this follows since by incentive compatibility $P_2$ cannot depend on $v_2$.
\end{proof}


The lematta above allow us to argue about payment, but don't tell us
which matching is chosen. This is
done in the following lemma.

\begin{lemma}
\label{lem:lower-3}
If $b_1 < b_2$, $b_1 < v_2$, and $ v_1 \neq v_2$, then any incentive compatible, Pareto-optimal, bidder rational and auctioneer rational combinatorial auction has to assign both items to agent $2$.
\end{lemma}
\begin{proof}
We want to show that independently of what agent $1$ says (but $v_1 \neq v_2$), agent $2$
will get both items.

We first concentrate on the case when $v_1 \le b_2$. Observe that the
only PO allocation assigns both items to
agent $2$. By Lemma~\ref{lem:lower-2}, if item $t_1$ was allocated to
agent $1$ then $P_2=0$. In this case
player $2$ can buy the item from $1$ and they are both better off.

Now, consider the case when $b_1 < v_1 < b_2$. By the above argument
player $1$ cannot be allocated item $t_1$. Suppose that for some value $v'_1 > b_1$ the allocation assigns item $t_1$ to agent $1$.
even though $v_2 > b_1$ and $b_2>b_1$. As agent $1$ is never charged more that her budget, $P_1 \leq b_1$. Then the utility for agent 1
is $v_1 - b_1 > 0$: agent 1 has incentive to lie about $v_1$, contradicting IC.

Hence, there is no value $v'_1 > b_1$ such that if agent 1 claims a valuation of $v'_1$ then the mechanism assigns $t_1$ to agent $1$.  This in turn implies that even if the truth is that $v_1 > b_2$, player $2$ must still be assigned both items $t_1$ and $t_2$.
\end{proof}

We are now ready to prove the main result of this section.

\begin{theorem} \label{thm:noprivateSapp}
There is no incentive compatible, Pareto-optimal, bidder rational and auctioneer rational combinatorial 
auction with public budgets, $b_a$, private valuations, $v_a$, and private sets of interest, $S_a$.
\end{theorem}
\begin{proof}
Consider the case of two agents, $1$ and $2$, and two items $t_1$,
$t_2$. Let $S_{1} = \{t_1, t_2\}$ and $S_{2} = \{ t_1, t_2\}$.
Additionally, Fix $v_1 = 10$, $v_2=11$, $b_1 = 4$ and $b_2 =5$. In
this case, by Theorem~\ref{thm:dnl-uniqueness}, the allocation must
coincide with the result of the dynamic clinching auction of \cite{DBLP:conf/focs/DobzinskiLN08}.

{\sl I.e.}, both agents
get one of the two items, $p_1 = 3$, and $p_2 = 2$. Without loss of
generality
assume that item $t_1$ is assigned to agent $1$ with probability at
least $\frac{1}{2}$ (if the mechanism is randomized).

Now, assume that the true set of interest for agent 1  was in fact $S_{1} = \{t_1\}$. We argue that agent 1 now has incentive to lie about $S_1$:
\begin{itemize}
\item if agent $1$ reports her true set of interest -- then by
Lemma~\ref{lem:lower-3} both items end up assigned to agent $2$, and
by Lemma~\ref{lem:lower-1} $P_1 = 0$, so her
utility is zero as well;
\item if agent $1$ lies and reports $\{t_1, t_2\}$ as her set of
interest -- then with probability $\leq \frac{1}{2}$ her utility
is equal to $0-3$, and with probability at least $\frac{1}{2}$ her
utility is equal to $10-3=7$, so on average his utility is at least
$-3\cdot \frac{1}{2} + 7\cdot \frac{1}{2} = 2$.
\end{itemize}
This concludes the proof as agent $1$ has incentive to lie in any incentive compatible, Pareto-optimal, bidder rational and auctioneer rational combinatorial auction.
\end{proof}
\end{document}